\pdfoutput=1

\documentclass[runningheads]{llncs}
\usepackage[T1]{fontenc}

\usepackage{cite}
\usepackage{amsmath,amssymb,amsfonts}
\usepackage{algorithm}
\usepackage[noend]{algorithmic}
\usepackage{graphicx}
\usepackage{textcomp}
\usepackage{cleveref}
\usepackage{xcolor}
\usepackage{comment}
\def\BibTeX{{\rm B\kern-.05em{\sc i\kern-.025em b}\kern-.08em
    T\kern-.1667em\lower.7ex\hbox{E}\kern-.125emX}}

\newtheorem{thm}{Theorem}

\renewenvironment{proof}{\noindent{\bf Proof:}}{\hspace*{\fill}\rule{6pt}{6pt}\medskip}
\newcommand{\bemph}[1]{\textbf{\textit{#1}}}

\usepackage[]{todonotes} 
\definecolor{okabe1}{HTML}{000000}
\definecolor{okabe2}{HTML}{E69F00}
\definecolor{okabe3}{HTML}{56B4E9}
\definecolor{okabe4}{HTML}{009E73}
\definecolor{okabe5}{HTML}{F0E442}
\definecolor{okabe6}{HTML}{0072B2}
\definecolor{okabe7}{HTML}{D55E00}
\definecolor{okabe8}{HTML}{CC79A7}
\newcommand{\claire}[2][inline]{\smallskip\todo[color=okabe8!50,#1]{\sf \textbf{Claire:} #2}}
\newcommand{\mike}[2][inline]{\smallskip\todo[color=okabe2!50,#1]{\sf \textbf{Mike:} #2}}
\newcommand{\gonzalo}[2][inline]{\smallskip\todo[color=okabe4!50,#1]{\sf \textbf{Gonzalo:} #2}}
\newcommand{\ham}{{\mathrm{hd}}}

\DeclareMathOperator{\polylog}{polylog}

\begin{document}
%
%
\pagestyle{plain}
%
%

\title{Improved Low-Overhead Communication-Efficient String Reconciliation
       and Edit Distance}

\author{Anonymous author(s)}

\author{Michael T. Goodrich\inst{1}\orcidID{0000-0002-8943-191X} \and
Gonzalo Navarro\inst{2}\orcidID{0000-0002-2286-741X} \and
Claire A. To\inst{1}\orcidID{0009-0008-9102-2219}}
%
%
\institute{University of California, Irvine, Irvine CA, USA
\and
Dept. of Computer Science, University of Chile, Santiago, Chile}

\maketitle

\begin{abstract}
Suppose two parties, Alice and Bob, hold long character strings, $X$ and $Y$, respectively, and they are interested in determining how similar $X$ and $Y$ are. {Moreover, they want to exchange the strings with cost proportional to their degree of dissimilarity.} Such problems arise, for example, in database and file system synchronization operations, as well
as in DNA sequence comparisons.
Since the strings are long, we are interested in methods that are communication-efficient and have low overhead in terms of the computations that Alice and Bob must perform, when the strings are similar enough.
In this paper, we provide a simple low-overhead communication-efficient algorithms for such string reconciliation and edit distance
problems, determining the edit distance $k$ between $X$ and~$Y$ using only $O(k\log^3 n)$ bits of communication
and $O(n\log k)$ time overhead, with high probability.


\end{abstract}


\section{Introduction}

A fundamental challenge in distributed computing lies in managing distributed data across communication networks. A problem arising frequently is that of checking or maintaining consistency across records in replicated databases, synchronizing files in peer-to-peer systems, and repairing or updating  databases across the network. Thus, there is a need for efficient distributed protocols for reconciling character strings, which can, e.g.,
represent files, Web pages, or DNA sequences, avoiding the communication overhead of retransmitting the whole strings.
As a concrete scenario, sequencing labs repeatedly generate genomes, and their clients then download them over the Internet to update genome libraries. In such scenarios, where clients already hold other genomes of the same species, which are known to be very similar to the new ones, clients and servers would like to minimize communication complexities as well as computational overheads. 


In some cases, version control systems can keep track of where the differences are, or both parties can share a common string, so that they can easily exchange only the needed edits to reconcile the strings. The scenarios we have described, however, may lack such shared string or edit information, so the parties {\em know} that the differences between their strings are small, but do not know {\em which} they are. Still, they want to find those differences, and reconcile the strings, within a communication cost that depends only on the amount of differences.


\subsection{Problem Statement}
Let $X = x_1 x_2 \ldots x_l$ and $Y = y_1 y_2 \ldots y_n$ be two strings held by two parties, Alice and Bob, respectively, 
where each $x_i, y_j \in \Sigma$, for $i=1,2,\ldots, l$ and $j=1,2,\ldots,n$, and $\Sigma$ is an alphabet whose size
may or may not depend on~$l$ or~$n$. 
In the case where $l=n$,
the {\em Hamming distance}, $\ham(X,Y)$, between $X$ and~$Y$ is 
the number of indices, $i$, where $x_i\not=y_i$, that is,
the number of character substitutions needed to transform $X$ into~$Y$
\cite{navarro2001guided,crochemore2007algorithms}.
A related notion that does not assume $l=n$
is the \bemph{edit distance}, also
known as {\em Levenshtein distance}, $L(X,Y)$,
between $X$ and~$Y$: it is 
the minimum number character insertions, deletions, or substitutions
needed to transform $X$ into~$Y$ (or vice versa).

In this paper, we are interested in distributed string
reconciliation algorithms with optimal (i.e., linear) computational 
overhead and small communication complexity 
for computing the edit distance between the two strings,
$X$ and $Y$, that are respectively held by Alice and Bob.
{This turns out to be equivalent to the problem of Alice and Bob exchanging their strings $X$ and $Y$ with communication cost based on their communicating a minimal set of edits between $X$ and $Y$.}
We assume that Alice and Bob can communicate over some channel but
they do not share any other computational
resources, such as CPUs or memory.
The \bemph{communication complexity} for a communication protocol for Alice and Bob is the number of bits that are communicated between them over the course of their distributed protocol, where we assume initially that neither party has knowledge of the other party's string.




Our results will hold \bemph{with high probability} (whp), meaning that they occur
with probability $1-1/n^c$, for some constant $c\ge1$. W.l.o.g.\ let us assume $l \le n$.
Our protocols use various independent sources of randomness, e.g., perfect hash functions, block boundary detection, chunk checksums, and IBLT hashing, each of which satisfies the whp guarantee.

\subsection{Our Contributions}
%
We present a simple algorithm that achieves a communication
complexity of {$O(k\log^3 n)$} bits
and $O(n\log k)$ time overhead, where $k$ is 
{the edit distance, 
with a probability of success over $1-1/n^c$ for any constant $c$}.

Our method involves a non-trivial use of the invertible Bloom lookup table (IBLT) data structure~\cite{goodrich2015invertiblebloomlookuptables,eppstein2011whatsthedifference,eppstein2010straggler}. This is an improvement on the IBLT-based technique introduced in our recent paper \cite{GNT26}, with which we share all the introductory sections for self-containedness.

\section{Related Work}

There has been considerable prior work on the general distributed edit distance problem, but we are not aware of 
any previous methods that are simple and communication-efficient while also having low computational overhead, which refers to the total number of local RAM operations performed by Alice and Bob, excluding time spent transmitting bits over the communication channel.
For example, Orlitsky~\cite{orlitsky1991interactive} gives a communication-optimal method that achieves
$O(k\log n)$ communication complexity, where $k$
is an upper bound on the edit distance between
two strings of size $\Theta(n)$, but requires computational overhead that is $n^{O(k)}$, i.e., exponential
in $k$.
The computational overhead for such protocols has been subsequently 
improved~\cite{jowhari2012efficient,chakraborty2016streaming,irmak2005improved,belazzougui2015efficient},
but the best methods still require $O(n\polylog(n))$ overhead and are still fairly complicated.
For instance, Chakraborty, Goldenberg, and Kouck{\`y}~\cite{chakraborty2016streaming} present a protocol 
with communication efficiency of $O(k^2\log n)$ bits and $O(n\polylog(n))$ overhead by embedding strings in 
a Hamming space using random walks and doing their communication in this Hamming space.
Belazzougui and Zhang~\cite{focs} achieve a communication complexity of $O(k(\log^2 k +\log n))$ bits
and $O(n\polylog(n))$ overhead assuming $k<n^{1/c}$ for a large constant $c$, but
their method is quite complex and their bounds have large constant factors.
In general, the previous methods for general distributed edit distance~\cite{jowhari2012efficient,chakraborty2016streaming,irmak2005improved,belazzougui2015efficient,focs} all
require $O(n\polylog(n))$ overhead, or worse, which in practice is unacceptable for the long strings that arise in various real-world settings. 

For specialized string reconciliation,
Agarwal, Chauhan, and Trachtenberg~\cite{puzzles}
present distributed methods for reconciling typical real-world strings. Their method has communication
complexity that scales linearly in the edit distance between the reconciling 
strings, as is also true for our methods.
Their reconstruction methods are based on using masks and shingling
to encode the strings and enumerating Eulerian paths to perform the
reconciliation.
As they admit, however, the number of such paths (and, hence, the
computational burden of such algorithms) can be exponentially large, hence, they do not achieve low overhead.
Like our approach, Kontorovich and Trachtenberg~\cite{Kontorovich}
reduce string reconciliation to 
the set reconciliation problem, but their approach still has suboptimal
performance. 
Song and Trachtenberg~\cite{song} revisit the approach using masks
and shingles with an improved Eulerian-path reconstruction method,
which they show emprically improves over the prior work
by Agarwal, Chauhan, and Trachtenberg~\cite{puzzles}, but the asymptotic
overhead of their method is unfortunately still not efficient.

\section{Invertible Bloom Lookup Table (IBLT)} \label{app:iblt}


IBLT  \cite{goodrich2015invertiblebloomlookuptables,eppstein2010straggler} is a space-efficient probabilistic data structure, based on Bloom filters, that supports various operations implemented by the XOR primitive. 
It can be used to solve the \textit{set reconciliation} problem~\cite{eppstein2011whatsthedifference}, which is the task of synchronizing two sets stored at two nodes across a communication link. IBLTs allow the parties to efficiently identify the items that are unique to each set so that both parties can update their local copies and obtain the other set as well. The main step in this process is to compute a \textit{set difference}, finding the set of elements that one party possesses but the other does not.

Suppose we are given a set $S=\{s_1, s_2, \ldots, s_n\}$, which we want to store probabilistically in an IBLT table, $T$, with $m$ cells.
We assume that we have an \bemph{encoding function}, $\mathsf{encode}$, which encodes each element, $s_i$, into an associated
unique, fixed-length binary string key, $e_i$.
We also have $\lambda$ hash functions that map any key to $\lambda$ distinct locations, and we let $\mathrm{HashToIndices}(e_i,\lambda,m)$ denote
the set of $\lambda$ locations determined by these hash functions for the key $e_i$.
Each IBLT cell, $T[j]$, stores a \texttt{(keysum, hashsum)} pair such that \texttt{keysum} contains the bitwise-XOR 
(denoted $\oplus$) of all encoded keys $e_i$ that are mapped to $T[j]$ by one of $T$'s hash functions; 
\texttt{hashsum} contains the  XOR of all fingerprints or hashes $H(e_i)$, where $H(e_i)$ is a $b$-bit secondary hash of $e_i$, such as a checksum, or cryptographic hash, used to verify or detect errors in the decoding process.

The IBLT supports both insertion and deletion of elements using XOR operations. Inserting an element corresponds to XOR-ing its encoded value into the table (at the $\lambda$ locations), while deleting it corresponds to {\em the same} operations.
%
%
%
To encode an entire set, $S$, we simply insert each of its elements.
%

For a set reconciliation protocol, Alice and Bob build respective tables, $T_A$ and $T_B$, by inserting all the elements of 
their sets, {$A$ and $B$, respectively, into initially zeroed IBLTs of the same size using the same hash functions, and exchange the tables.
From these two IBLTs, $T_A$ and $T_B$, both Alice and Bob can compute an IBLT representing 
the symmetric difference between $A$ and $B$.} This simple method
is shown in Algorithm~\ref{alg:subtract}.
The fact that this algorithm computes an IBLT representation of the symmetric difference of $A$ and $B$ follows
immediately from the fact that $x\oplus x=0$ for any value $x$.

\begin{algorithm}[b]
\caption{IBLT Subtract ($T = T_A \oplus T_B$)}
\begin{algorithmic}[1]
\FOR{$i=0$ to $m-1$}
\STATE $T[i].\texttt{keysum} \gets T_A[i].\texttt{keysum} \oplus T_B[i].\texttt{keysum}$
\STATE $T[i].\texttt{hashsum} \gets T_A[i].\texttt{hashsum} \oplus T_B[i].\texttt{hashsum}$
\ENDFOR
\end{algorithmic}
\label{alg:subtract}
\end{algorithm}

\subsubsection{Listing Set Entries}
If an IBLT is not ``too'' full, we can list out its entries.
The listing of entries or decoding is a destructive $O(m)$-time procedure that recovers all keys in the IBLT or reports
that the IBLT is too full to achieve this successfully.
Say that a cell $T[i]$ is \bemph{pure} if
\[
H(T[i].\texttt{keysum}) = T[i].\texttt{hashsum},
\]
which clearly holds for a cell holding exactly one key but could admittedly also hold if there is a checksum
collision. 
Note that
if a cell is pure, we can recover its key by extracting the value from $T[i].\texttt{keysum}$ and then deleting the element
from the table, which is also referred to as \bemph{peeling} they key from the corresponding $\lambda$ hashed cells. 
This may make other cells pure, which can eventually allow us to recover all the values stored in the IBLT.
Algorithm~\ref{alg:decode} shows this process.

\begin{algorithm}[t]
\caption{IBLT Decode $T$}
\begin{algorithmic}[1]
\WHILE{$\exists~i$, s.t. $H(T[i].\texttt{keysum}) = T[i].\texttt{hashsum}$}
    \STATE $s\gets \mathsf{decode}(T[i].\texttt{keysum})$
    \STATE Output $s$
    \STATE delete$(T,s)$
\ENDWHILE

\FOR{$i=0$ to $m-1$}
    \IF{$T[j].\texttt{keysum} \neq 0$ OR $T[j].\texttt{hashsum} \neq 0$}
    \STATE \textbf{return} FAIL
    \ENDIF
\ENDFOR
\STATE \textbf{return} SUCCESS
\end{algorithmic}
\label{alg:decode}
\end{algorithm}


\begin{lemma} \label{lem:purity1}
    The probability that a cell containing $t \ge 2$ keys produces a false positive on the purity test is $2^{-b}$.
\end{lemma}
\begin{proof} 
    Let $H$ be a uniform random hash function that produces $b$-bit hash values.
    Suppose a cell contains $t \ge 2$ keys, with \texttt{keysum} $=x_1 \oplus x_2 \oplus \ldots \oplus x_t$ and \texttt{hashsum} $=H(x_1) \oplus H(x_2) \oplus \ldots \oplus H(x_t)$. Define a random variable $Z=H(\texttt{keysum})\oplus \texttt{hashsum}=H(x_1 \oplus \ldots \oplus x_t)\oplus H(x_1) \oplus \ldots \oplus H(x_t)$. Because $H$ is uniform and random, $H(x_1 \oplus \dots \oplus x_t)$ is independent of $H(x_1), \dots, H(x_t)$ and each value is uniformly distributed over $\{0,1\}^b$. Thus, the probability of failure is $2^{-b}$.
\end{proof}

\begin{corollary} \label{cor:purity2}
    In an IBLT of $m$ cells, the probability that the purity test yields a false positive for any cell is at most $m 2^{-b}$.
\end{corollary}

\begin{proof} 
    Apply the union bound to the failure probability from Lemma~\ref{lem:purity1}.
\end{proof}

In our setting, we want to recover a symmetric difference of size at most $d$ with high probability, even if $d$ is a constant.
We use an IBLT of size $m = c \cdot d\log n$ for a $c\ge 1$, and $\lambda$
being $\Theta(\log n)$, so that by
an analysis similar to that of Goodrich, Kitagawa, and Mitzenmacher~\cite{goodrich2025parallel}, when the current number of elements in the IBLT is at most $d$, we can successfully list all entries of $T$ whp:

\begin{thm}
\label{thm:whp}
    Let $T$ be an IBLT with  $m=cd\log n$ cells and 
    $\lambda=(c/2)\log n$ hash functions, for $c \ge 4$. 
    If $T$ stores at most $d$ elements, then it can be successfully decoded by the peeling algorithm with 
    probability at least $1-1/n^{c/2-1}$.
\end{thm}

\begin{proof} 
    Let $X_{i}$ be a random variable that is $1$ if the $i$-th element has no pure cell and is $0$ otherwise.
    Since there are at most $d$ elements in $T$ and it has $m=cd\log n$ cells and 
    $\lambda=(c/2)\log n$ hash functions, at least half the cells in the IBLT are empty. Thus, any one of the $\lambda$ hash functions will map an element
    to what would otherwise be an empty cell with probability at least $1/2$.
    Since the $\lambda$ hash functions are independent, this implies that
    \[
    \Pr(X_i=1) \le \frac{1}{2^\lambda} = \frac{1}{n^{c/2}}.
    \]
    Thus, by a union bound, the probability that every element in $T$ has a pure cell is at least
    $1-1/n^{c/2-1}$.
\end{proof}

To sum up, suppose two sets, $A$ and $B$, of size $O(n)$, are held by Alice and Bob respectively, which are assumed to differ by at most $d$ elements. Each party independently constructs an IBLT, $T_A$ and $T_B$ of $m = O(d\log n)$ cells, and inserts all elements into the table. Even if either $T_A$ and $T_B$ is too full to decode individually, there will be pure cells in $T_A \oplus T_B$, i.e., containing only one key. That is, with properly chosen parameters, Alice and Bob can reconcile their sets whp, in $O(n)$ time and $O(d\beta\log n)$ bits of communication, where $\beta>0$ is the size of the keys, by exchanging their tables $T_A$ and $T_B$ and performing the decoding algorithm on $T_A \oplus T_B$. 

\newcommand{\LCP}{\mathrm{LCP}}

\section{Our Algorithm}

\bemph{Locally consistent parsing (LCP)} \cite{coin_tossing,locally_1,MSU97} divides a string, $X=x_1x_2\ldots x_n$, into a sequence, $\mathcal{C}$, of contiguous substrings, called \bemph{chunks},
$C_1$, $C_2$, $\ldots$, $C_r$, so that $X=C_1 || C_2 || \cdots || C_{r}$, where $||$ denotes concatenation.
The version of LCP we use \cite{MSU97,CEKNPtalg20} achieves this by defining a kind of random permutation $\pi$ on the alphabet symbols of $X$, and dividing the string, $X$, by a simple algorithm that ends a chunk wherever a {\em local minimum} occurs. Previously to this chunking, the technique groups {\em runs} of equal symbols. Precisely
\begin{enumerate}
    \item Every maximal sequence of two or more equal consecutive symbols, $x_i x_{i+1} \cdots$ $x_j = a^{j-i+1}$, for $j>i$, is replaced in $X$ by a new symbol $\langle a,j-i+1\rangle$. Let $X' = x_1'x_2'\ldots x_{n'}'$ be the resulting sequence.
    \item A random injective function $\pi$ is defined from the alphabet of $X'$ to integers in $[0,n^{c+2}-1]$ for some constant $c\ge 1$. Crucially, $\pi(\langle a,k\rangle)=\pi(a)$ for any $k$. Also, for technical convenience, we fix $\pi(x'_1)=n^{c+2}-1$ and $\pi(x'_{n'})=0$.
    \item The string, $X'$, is divided into a sequence of chunks: a chunk ends at every position $i$ such that $\pi(x_{i-1}') > \pi(x_i') < \pi(x_{i+1}')$; the last chunk ends at position $n'$.
    \item The resulting sequence of chunks is called $\LCP(X)$.
\end{enumerate}

Note that, since there cannot be two consecutive runs or symbols with the same symbol $c$, it holds $\pi(x_i') \neq \pi(x'_{i+1})$ for every $i$. Further, there cannot be chunks of length $1$, as we show next.

\begin{lemma} \label{lem:length2}
Every chunk is of length at least $2$.
\end{lemma}
\begin{proof}
Follows from the fact that there cannot be two consecutive minima by definition, and that by our choice of $\pi(x_1')$ and $\pi(x_{n'}')$, the first and last chunks cannot be of length $1$ either.
\end{proof}

The next lemma shows, on the other hand, that the chunks are short whp.

\begin{lemma}
With high probability, the longest chunk of $X'$ is of length $O(\log n)$.
\end{lemma}
\begin{proof}
Consider the distance between two local minima, $i$ and $j$, in $\pi(X')$. The values of $\pi$ in this range must have an increasing phase starting at $i$, of length $\ell^+$, and a decreasing phase ending at $j$, of length $\ell^-$. Since there cannot be repeated symbols $X'[i]$ in the increasing nor in the decreasing phase, the values of $\pi$ within each phase are independent. Therefore, the probability of lengths $\ell^+$ or $\ell^-$ exceeding $d$ is at most $2/2^d$. By the union bound, since there are at most $|X'|/2$ chunks in $X'$, the probability of some chunk exceeding length $2d$ is at most $|X'|/2^d \le n/2^d$. By choosing $d \ge c \log_2 n$, this is at most $1/n^c$.
\end{proof}

We implement $\pi$ as a perfect hash function chosen uniformly at random from a 2-universal family, such as $\mathcal{H} = \{ h_{\alpha,\beta}(x) = 1+[((\alpha x+\beta) \bmod p) \bmod (n^{c+2}-1)],~ \alpha,\beta \in \mathbb{N}, \alpha>0\}$ for a prime $n^{c+2} < p = O(n^{c+2})$ (fixing $\pi(x_1')=0$ and $\pi(x'_{n'})=n^{c+2}-1$). Any random choice $\pi = h_{\alpha,\beta}$ is not perfect with probability $O(n^2/n^{c+2}) = O(1/n^c)$, so we find a suitable $\pi$ in $O(1)$ attempts whp. Note $\pi$ has an $O(\log n)$-bit description. Further, we iterate in the quest for $\pi$ until all the chunks it produces are of length at most $2c \log_2 n$. Overall, we build $\pi$ in $O(n)$ time whp.

The general process involves consecutive {\em rounds} of LCP: calling $X_0 = X$ the original sequence, we apply again LCP on the sequence $X_1 = \LCP(X_0)$, $X_2 = \LCP(X_1)$, and so on until reaching some $X_r$ such that $|X_r|=1$. Precisely:
\begin{enumerate}
    \item We produce $Y_0 = \LCP(X_0)$ as described above, so $Y_0$ is a sequence of chunks. Let $\pi_0$ be the function used for chunking; chunks are of logarithmic length.
    \item We give {\em integer names} to the distinct chunks in $Y_0$. Those names are given with a $\Theta(\log n)$-bit hash function $h_0$ that is perfect for the chunks in $Y_0$, so $N = h_0(C)$ is the name of chunk $C$. 
    \item The string $X_1$ is formed by replacing the chunks in $Y_0$ by their names.
    \item We repeat the process, forming $Y_1 = \LCP(X_1)$ with a new random function $\pi_1$ and hash function $h_1$, and so on.
\end{enumerate}

The names will be recorded, for all the rounds $i$, in sets
\[ A_i = \{ N \to C,~ C \text{ is a chunk of } Y_i \text { assigned name } N = h_i(C) \}. \]  
We now show that all Alice needs to send Bob so he can reconstruct $X$ are, essentially, the sets $A_i$ and the hash functions $\pi_i$ and $h_i$.

\begin{lemma} \label{lem:reconstruct}
Given the sets $A_i$, the hash functions $\pi_i$ and $h_i$, the integer $r$, and the string $X_r$ of length 1, we can reconstruct $X$.
\end{lemma}
\begin{proof}
We proceed by induction, showing that we can reconstruct every $X_i$, from $i=r$ to $i=0$, where we finish because $X_0 = X$. For the base case, $X_r$ is given. Now, assume we have decoded $X_{i+1}$. Using $A_i$, we replace every symbol $X_{i+1}[j] = N$ by the chunk $C$, where $N \to C$ appears in $A_i$. Further, any symbol of the form $\langle c,r\rangle$ in $C$ is rewritten as $r$ copies of $c$. The result is $X_i$.
\end{proof}

Finally, we show that the total size of the sets is linear. 

\begin{lemma}
The total size of all the sets $A_i$ and hash functions $\pi_i$ and $h_i$ is $O(n)$, and there are $r \le \log_2 n$ levels. 
\end{lemma}
\begin{proof}
Collapsing the runs of $X_i$ to form $X_i'$ cannot increase its length, so $|X_i'| \le |X_i|$. Further, Lemma~\ref{lem:length2} implies $|X_{i+1}| \le |X_i'|/2$. Therefore, we obtain $|X_r|=1$ after $r \le \log_2 n$ rounds. There are then $O(\log n)$ functions $\pi_i$ and $h_i$. 

By connecting the names created at level $i$ with the places where they are used in chunks of level $i+1$, the sets $A_i$ can be regarded as a tree with root $X_r[1]$ and leaves forming $X$ (concretely, the union of the sets $A_i$ can be seen as the rules of a context-free grammar that generates $X$). Since there are no unary nodes in this tree, it has $O(n)$ nodes. The sets $A_i$ add up to less because they represent each name only once, thus $\sum_{i=0}^r |A_i| = O(n)$.
\end{proof}

In particular, this implies that the CPU cost to reconstruct the sets is $O(n)$, in addition to the $O(n)$ time needed whp to build suitable hashes.
We now state the key result that allows us using LCP for string reconciliation.

\begin{lemma} \label{lem:klogn}
A single edit in $X$ can alter at most $4\log_2 n$ elements in the corresponding sets $A_i$.
\end{lemma}
\begin{proof}
We prove in Appendix~\ref{app:maxalt} that an edit operation in $X_0$ cannot propagate to more than 4 chunk edits (i.e., chunks inserted, deleted, or modified) in any of the sequences $X_i$. Therefore, no more than $4r \le 4\log_2 n$ rules can be altered.
\end{proof}

Let us assume that we have an upper-bound estimate, $k\ge1$, 
for the edit distance between the strings, $X$ and $Y$, held respectively by Alice and Bob (we later lift this assumption). 
Let us focus on Bob obtaining the string $X$ from Alice; the process for Alice obtaining Bob's $Y$ is analogous.

We first obtain common functions $\pi_i$ and $h_i$ for Alice and Bob. This can be achieved asymmetrically, as follows: Alice generates suitable functions and sends them to Bob, using only $O(\log^2 n)$ bits of communication. Bob checks that the functions are suitable for the construction of his own sets $B_i$ from $Y$.\footnote{To make them compatible, we take $n=\max(|X|,|Y|)$ when building the functions, and the domain of the injective functions is sufficient to map any integer name.} If not, he requests a new set of functions to Alice. Only $O(1)$ iterations suffice whp. Note that, because these functions are shared, equal chunks $C$ will receive equal names $N = h_i(C)$ and functions $\pi_i$ produce the same cuts, so Lemma~\ref{lem:klogn} implies that there will be at most $4k$ differing elements between the sets $A_i$ and $B_i$. 

Alice creates an IBLT of appropriate size to allow Bob obtaining the elements of each set $A_i$ with probability at least $1-1/\eta$, for
our desired confidence parameter of successful decoding, $\eta\le n^c$. 
Since there will be at most $4k$ differing elements between $A_i$ and $B_i$, Alice builds IBLTs of
$O(k\log \eta)$ cells, with each cell storing a \texttt{keysum} and \texttt{hashsum} of $\Theta(\log n)$ bits.
Alice then transmits $r$, $X_r$, and $r$ IBLTs encoding her sets $A_i$, from $A_{r-1}$ to $A_0$, in rounds. 

Bob will proceed by rounds, obtaining the strings $X_i$ from $X_r$ to $X_0=X$, by following Lemma~\ref{lem:reconstruct}. The string $X_r$ is sent explicitly by Alice, which gives the base of the induction. In general, Bob has already decoded $X_{i+1}$ and wants to decode $X_i$ from $X_{i+1}$ and $A_i$. Since there are at most $4k$ chunk contents $C$ that differ between $A_i$ and $B_i$, the IBLT will allow Bob recovering $A_i$ from $B_i$ whp. Precisely, Bob will identify, whp, the at most $4k$ elements in $A_i \setminus B_i$, which he will request from Alice; she will send him back the explicit chunks. Each such chunk requires $O(\log^2 n)$ bits of communication. With $X_{i+1}$ and $A_i$, Bob will then reconstruct $X_i$ using Lemma~\ref{lem:reconstruct}. The total failure probability is $O((1/\eta)\log n)$, so we still succeed whp if $\eta > n\log n$.

The $r$ IBLTs add up to $O(k\log n \log \eta)$ cells, and each cell requires $O(\log n)$ bits because the chunks are of constant size and the alphabet and names requires $O(\log n)$ bits each. On the other hand, the exchanged chunks add up to $O(k \log n)$, and each requires $O(\log^2 n)$ bits of communication.
Therefore, we have the following result, which yields $O(k\log^3 n)$-bit communication complexity in order to succeed whp.

\begin{thm}
\label{thm:prob}
    Let $k$ be an {a-priori-known} upper bound on
    the edit distance between two strings, $X$ and $Y$, of length $\Theta(n)$ and defined over an alphabet of size $O(n)$,
held by two separate parties, Alice and Bob. 
    Then our IBLT-chunking algorithm allows Alice and Bob to determine the 
other party's string in $O(n)$ time 
    using 
    \[ O(k\log^2 n (\log n + \log \eta)) ~=~ O(k\log^3 n)\]
bits of communication, all with probability at least $1-1/\eta$ for any $n\log n < \eta \le n^c$ for any desired constant $c > 1$.
\end{thm}

 {Once both Alice and Bob have the chunks of both strings, they can compute their edit distance in $O(n+k^2)$ additional time \cite{landau1998incremental}.} 
Thus, when $k = O(\sqrt{n})$, 
we can also compute the precise edit distance within this cost.

To conclude, we note that we can obtain a similar result {\em without} knowing an \textit{a priori} bound $k$ on the edit distance, by slightly increasing the time overhead.

\begin{corollary}
    {The same result of Theorem~\ref{thm:prob}  can be obtained if $k$ is the actual edit distance between $X$ and $Y$, a priori unknown, with $O(n\log k)$ time overhead.}
\end{corollary}
\begin{proof}
{Apply Theorem~\ref{thm:prob} with upper bound $k'=1,2,4,8,\ldots$ until it succeeds. With probability at least $1-1/\eta$, the try with $k \le k' < 2k$ will succeed, and each party will obtain the other's string. The total communication cost is as in Theorem~\ref{thm:prob}, replacing $k$ by $\sum_{j=0}^{\lceil \log_2 k\rceil} 2^j = O(k)$, and the time overhead is $O(n\log k)$ because we process the chunks $O(\log k)$ times, once per value of $k'$.}
\end{proof}

\section{Conclusions}

We have introduced a new IBLT-based communication efficient algorithm for string reconciliation between two parties that separately store two strings of size $\Theta(n)$, improving our recent result also based on IBLT \cite{GNT26} by exchanging only $O(k\log^3 n)$ bits whp.

Our technique is based on generating a randomized grammar that outputs the string, and exchanging the grammars. We show that each edit alters only $O(\log n)$ nonterminals in the grammar, if a specific grammar construction \cite{CEKNPtalg20} is used. Our technique has the additional benefit that the grammar it uses is guaranteed to generate a representation is of size $O(\gamma\log(n/\gamma))$, where $\gamma$ is the size of the smallest attractor \cite{KP18} of the string to exchange. For example, $\gamma$ lower-bounds the size of the Lempel-Ziv parse \cite{LZ76} of the string. If the texts Alice and Bob have to exchange are very repetitive (e.g., large collections of genomes of the same species, as in genomic databases), their $\gamma$ value is much less than their collection size $n$ \cite{Nav20}. In such cases, they may have the texts directly compressed in grammar form, and can reconcile their databases without spending time or space proportional to $n$, but to $\gamma$. By using a slightly more complex grammar \cite{KNOalgo23}, we can obtain the same result with a grammar of size $O(\delta\log(n/\delta))$, where $\delta \le \gamma$ is the normalized substring complexity of the strings \cite{KNPtit22}.

\bibliographystyle{IEEEtranS}
\bibliography{ref}

@article{goodrich2015invertiblebloomlookuptables,
      title={Invertible {Bloom} Lookup Tables}, 
      author={Michael T. Goodrich and Michael Mitzenmacher},
      year={2015},
      volume={1101.2245},
      journal = "CoRR",
      !archivePrefix={arXiv},
      !primaryClass={cs.DS},
      !url={https://arxiv.org/abs/1101.2245}, 
}

@ARTICLE {
        Nav20,
        TITLE = "Indexing Highly Repetitive String Collections, 
		 Part {I}: Repetitiveness Measures",
        AUTHOR = "Gonzalo Navarro",
        JOURNAL = {ACM Computing Surveys},
        YEAR = 2021,
        VOLUME = 54,
        NUMBER = 2,
        PAGES = "article 29"
        }

@article{eppstein2011whatsthedifference,
author = {Eppstein, David and Goodrich, Michael T. and Uyeda, Frank and Varghese, George},
title = {What's the Difference? {E}fficient Set Reconciliation Without Prior Context},
year = {2011},
_publisher = {ACM},
_address = {New York, NY, USA},
volume = {41},
number = {4},
issn = {0146-4833},
_url = {https://doi.org/10.1145/2043164.2018462},
_doi = {10.1145/2043164.2018462},
abstract = {We describe a synopsis structure, the Difference Digest, that allows two nodes to compute the elements belonging to the set difference in a single round with communication overhead proportional to the size of the difference times the logarithm of the keyspace. While set reconciliation can be done efficiently using logs, logs require overhead for every update and scale poorly when multiple users are to be reconciled. By contrast, our abstraction assumes no prior context and is useful in networking and distributed systems applications such as trading blocks in a peer-to-peer network, and synchronizing link-state databases after a partition.Our basic set-reconciliation method has a similarity with the peeling algorithm used in Tornado codes [6], which is not surprising, as there is an intimate connection between set difference and coding. Beyond set reconciliation, an essential component in our Difference Digest is a new estimator for the size of the set difference that outperforms min-wise sketches [3] for small set differences.Our experiments show that the Difference Digest is more efficient than prior approaches such as Approximate Reconciliation Trees [5] and Characteristic Polynomial Interpolation [17]. We use Difference Digests to implement a generic KeyDiff service in Linux that runs over TCP and returns the sets of keys that differ between machines.},
journal = {SIGCOMM Computer Communication Review},
month = aug,
pages = {218–229},
numpages = {12},
keywords = {set difference, invertible bloom filter, difference digest}
}

@ARTICLE{puzzles,
  author={Agarwal, S. and Chauhan, V. and Trachtenberg, A.},
  journal={IEEE Transactions on Parallel and Distributed Systems}, 
  title={Bandwidth Efficient String Reconciliation Using Puzzles}, 
  year={2006},
  volume={17},
  number={11},
  pages={1217-1225},
  doi={10.1109/TPDS.2006.148}}

@INPROCEEDINGS{Kontorovich,
  author={Kontorovich, Aryeh and Trachtenberg, Ari},
  booktitle={Proc. IEEE Int. Symp. on Information Theory}, 
  title={String Reconciliation with Unknown Edit Distance}, 
  year={2012},
  volume={},
  number={},
  pages={2751-2755},
  doi={10.1109/ISIT.2012.6284024}}

@INPROCEEDINGS{song,
  author={Song, Bowen and Trachtenberg, Ari},
  booktitle={Proc. 57th Allerton Conf. on Communication, Control, and Computing (Allerton)}, 
  title={Scalable String Reconciliation by Recursive Content-Dependent Shingling}, 
  year={2019},
  volume={},
  number={},
  pages={623-630},
  doi={10.1109/ALLERTON.2019.8919901}}

@article{navarro2001guided,
  title={A Guided Tour to Approximate String Matching},
  author={Navarro, Gonzalo},
  journal={ACM Computing Surveys},
  volume={33},
  number={1},
  pages={31--88},
  year={2001},
  publisher={ACM},
  address={New York, NY, USA}
}

@book{crochemore2007algorithms,
  title={Algorithms on Strings},
  author={Crochemore, Maxime and Hancart, Christophe and Lecroq, Thierry},
  year={2007},
  publisher={Cambridge University Press}
}

@ARTICLE 
        { CEKNPtalg20,
          AUTHOR = "A. R. Christiansen and M. B. Ettienne and T. Kociumaka
		    and G. Navarro and N. Prezza",
	  TITLE = "Optimal-Time Dictionary-Compressed Indexes",
	  JOURNAL = "ACM Transactions on Algorithms",
          YEAR = 2020,
	  VOLUME = 17,
	  NUMBER = 1,
	  PAGES = "article 8"
        }

@ARTICLE
        { KNOalgo23,
	  TITLE = "Near-Optimal Search Time in $\delta$-Optimal Space, and Vice Versa",
          AUTHOR = "T. Kociumaka and G. Navarro and F. Olivares",
	  JOURNAL = "Algorithmica",
	  VOLUME = 86,
	  PAGES = "1031--1056",
	  YEAR = 2024
	}

@ARTICLE
        { KNPtit22,
	  TITLE = "Toward a Definitive Compressibility Measure
		   for Repetitive Sequences",
          AUTHOR = "T. Kociumaka and G. Navarro and N. Prezza",
	  JOURNAL = "IEEE Transactions on Information Theory",
	  YEAR = 2023,
	  VOLUME = 69,
	  NUMBER = 4,
	  PAGES = "2074--2092"
	}

@inproceedings { KP18,
        author = "D. Kempa and N. Prezza",
        title = "At the Roots of Dictionary Compression: String Attractors",
        booktitle = "Proc. 50th Annual ACM Symposium on the Theory of
                        Computing (STOC)",
        !booktitle = "Proc. 50th STOC",
        year = 2018,
        pages = "827--840" 
        }

@ARTICLE{LZ76,
        author = {A. Lempel and J. Ziv},
        journal = {IEEE Transactions on Information Theory},
        number = {1},
        pages = {75--81},
        title = {On the Complexity of Finite Sequences},
        volume = {22},
        year = {1976}
}

@Article{MSU97,
  author    = {Kurt Mehlhorn and R. Sundar and Christian Uhrig},
  title     = {Maintaining Dynamic Sequences under Equality Tests in Polylogarithmic Time}, 
  journal   = {Algorithmica},
  year      = {1997},
  volume    = {17},
  number    = {2},
  pages     = {183--198},
  !bibsource = {dblp computer science bibliography, https://dblp.org},
  !doi       = {10.1007/BF02522825},
  !groups    = {2018_edandpartition, thesis, 2019_index}, 
}

@InProceedings{coin_tossing,
  author     = {Cole, R and Vishkin, U},
  booktitle  = {18th Annual ACM Symposium on Theory of Computing (STOC)},
  title      = {Deterministic Coin Tossing and Accelerating Cascades: Micro and Macro Techniques for Designing Parallel Algorithms},
  year       = {1986},
  pages      = {206--219},
  _address   = {New York, NY, USA},
  _location  = {Berkeley, California, USA},
  _numpages  = {14},
  _publisher = {Association for Computing Machinery},
  _series    = {STOC '86},
  doi        = {10.1145/12130.12151},
  isbn       = {0897911938},
}

@InProceedings{locally_1,
  author     = {Tuğkan Batu and S{\"{u}}leyman Cenk Sahinalp},
  booktitle  = {9th International Conference on Developments in Language Theory},
  title      = {Locally Consistent Parsing and Applications to Approximate String Comparisons},
  year       = {2005},
  pages      = {22--35},
  !series     = {LNCS},
  !volume     = {3572},
  _editor    = {Clelia de Felice and Antonio Restivo},
  _publisher = {Springer}, 
  !doi        = {10.1007/11505877_3},
  isbn       = {978-3-540-31682-4},
}

@article{landau1998incremental,
  title={Incremental String Comparison},
  author={Landau, Gad M and Myers, Eugene W and Schmidt, Jeanette P},
  journal={SIAM Journal on Computing},
  volume={27},
  number={2},
  pages={557--582},
  year={1998},
  publisher={Society for Industrial and Applied Mathematics}
}

@INPROCEEDINGS{focs,
  author={Belazzougui, Djamal and Zhang, Qin},
  booktitle={Proc. IEEE Symp. on Foundations of Computer Science (FOCS)}, 
  title={Edit Distance: Sketching, Streaming, and Document Exchange}, 
  year={2016},
  pages={51-60},
  doi={10.1109/FOCS.2016.15}}

@inproceedings{orlitsky1991interactive,
  title={Interactive Communication: Balanced Distributions, Correlated Files, and Average-Case Complexity},
  author={Orlitsky, Alon},
  booktitle={Proc. IEEE Symp. on Foundations of Computer Science (FOCS)}, 
  pages={228--238},
  year={1991}
}

@inproceedings{jowhari2012efficient,
  title={Efficient Communication Protocols for Deciding Edit Distance},
  author={Jowhari, Hossein},
  booktitle={Proc. European Symp. on Algorithms (ESA)},
  pages={648--658},
  year={2012},
  _organization={Springer}
}

@inproceedings{chakraborty2016streaming,
  title={Streaming Algorithms for Embedding and Computing Edit Distance in the Low Distance Regime},
  author={Chakraborty, Diptarka and Goldenberg, Elazar and Kouck{\`y}, Michal},
  booktitle={Proc. ACM Symp. on Theory of Computing (STOC)},
  pages={712--725},
  year={2016}
}

@inproceedings{irmak2005improved,
  title={Improved Single-Round Protocols for Remote File Synchronization},
  author={Irmak, Utku and Mihaylov, Svilen and Suel, Torsten},
  booktitle={Proc. 24th Joint Conf. of the IEEE Computer and Communications Societies},
  volume={3},
  pages={1665--1676},
  year={2005},
}

@article{belazzougui2015efficient,
  title={Efficient Deterministic Single Round Document Exchange for Edit Distance},
  author={Belazzougui, Djamal},
  journal={CoRR},
  volume = {1511.09229},
  year={2015}
}

@article{eppstein2010straggler,
  title={Straggler identification in round-trip data streams via Newton's identities and invertible {Bloom} filters},
  author={Eppstein, David and Goodrich, Michael T},
  journal={IEEE Transactions on Knowledge and Data Engineering},
  volume={23},
  number={2},
  pages={297--306},
  year={2010},
  publisher={IEEE}
}

@inproceedings{goodrich2025parallel,
  title={Parallel Peeling of Invertible {Bloom} Lookup Tables in a Constant Number of Rounds},
  author={Goodrich, Michael T and Kitagawa, Ryuto and Mitzenmacher, Michael},
  booktitle={Proc. Int. Conf. on Current Trends in Theory and Practice of Computer Science},
  pages={70--84},
  year={2025},
  !organization={Springer}
}

@article { GNT26,
    title = "Simple Low-Overhead Communication-Efficient
String Reconciliation and Edit Distance",
    author = "M. T. Goordrich and G. Navarro and C. A. To",
    journal = "CoRR",
    volume = "2608.19149",
    year = 2026}

\appendix

\section{Maximum Number of Affected Chunks} \label{app:maxalt}

We prove that a single edit in the base sequence $X_0$ (i.e., an insertion, a deletion, or a substitution of a symbol) cannot induce more than $4$ edits in any subsequent sequence $X_i$. First, we note that this is tight: if we assume that $\pi$ respects lexicographic order between symbols, then substituting, within $X_0$, $a$ by $d$ in $cb\underline{a}|bc$ produces $cb|\underline{d}b|c$ (we underline the changed symbols and indicate chunk limits with a vertical bar). This changes two chunks to three, which amounts to two chunk substitutions and one chunk insertion, that is, $3$ edits in $X_1$. Once we have $3$ consecutive edits, we can produce $4$: consider $dc\underline{ba}|bc$ and replace $ba$ by $dcd$ (with two substitutions and one insertion). This yields $dc|\underline{dc}|\underline{d}b|c$, converting two chunks into four  (i.e., two insertions and two substitutions). We now show that four is the limit: 4 edits in $X_i$ produce at most 4 edits in $X_{i+1}$.

Let us first disregard the possibility of runs, and consider 4 substitutions. We say that these affect a {\em window} $W$ of 4 symbols, which becomes a new window $W'$, also of 4 symbols. We define the area of influence of those edits as follows.

\begin{definition}
Let the edits in a string apply on a window $W$. The corresponding {\em extended window} is the context $x_1 W y_1 y_2$ formed by extending $W$ one symbol to the left and two to the right.    
\end{definition}

Note that changes in $W$ can affect chunks it does not intersect, as in $cb|\underline{cdcb}$ becoming $cb\underline{a}|\underline{bcd}$ if we replace $W=cdcb$ by $W'=abcd$, but they cannot affect chunks out of the extended window, because those chunks are fully defined by symbols out of $W$, as proved next.

\begin{lemma} \label{lem:sourceW}
Chunks not intersecting the extended window of $W$ cannot be altered by the edits in $W$.
\end{lemma}
\begin{proof}
Let a chunk finish before the extended window $x_1 W y_1 y_2$, that is, (in the worst case) extending up to $x_2$ in the context $x_3 x_2 x_1 W y_1 y_2$. Its ending position depends on the local minimality of $x_2$, which depends on $\pi(x_3) \pi(x_2) \pi(x_1)$, but not on $W$. Similarly, let a chunk start right after the extended window, (in the worst case) at $y_3$ in the context $x_1 W y_1 y_2 y_3$. Its starting position depends on the minimality of $y_2$, which depends on $\pi(y_1) \pi(y_2) \pi(y_3)$, but not on $W$.  
\end{proof}

Since local minima must be separated by at least one other symbol, and the extended window is of length 7, it cannot intersect more than 4 consecutive chunks. Further, the edits can leave at most 3 local minima within the extended window, as proved next.

\begin{lemma} \label{lem:targetW}
There can be at at most 3 local minima in $x_1 W' y_1 y_2$.
\end{lemma}
\begin{proof}
The edits in $W'$ might affect the local minimality of $x_1$ and of $y_1$, but not of $y_2$ because that depends on $\pi(y_1)\pi(y_2)\pi(y_3)$ in the context $x_1 W' y_1 y_2 y_3$. Since $|x_1 W' y_1| \le 6$, it can contain at most 3 local minima.
\end{proof}

Those 3 local minima induce 3 consecutive chunk endpoints, which affect at most 4 consecutive chunks. Those are obtained from the 4 chunks intersected by the extended window via at most 4 chunk edits (insertions, deletions, or replacement of chunks).

Assume now that our 4 edit operations include $i \le 4$ insertions. Then the window $W$ is of length $4-i$, and the extended window is of length $7-i$. After the edits, $W'$ is of length 4 and its extended window of length 7. If the edits include $d$ deletions, then the window $W'$ after the edits is of length $4-d$, and its extended window of length $7-d$. In all cases, the edits convert at most 4 chunks into other 4 chunks, which can be obtained with at most 4 chunk edits.

Finally, let us consider the runs. Four edits strictly within a run $\langle a,k\rangle$ produce $\langle a,k_1\rangle W' \langle a,k_2\rangle$ for some window $W'$ of length at most 4, and some $k_1$ and $k_2$. Assume we temporarily replace the original $\langle a,k\rangle$ by $\langle a,k_1\rangle \langle a,k_2\rangle$; if $\langle a,k\rangle$ ended a chunk we assume the chunk ends at $\langle a,k_2\rangle$, so our change affects only the chunk where $\langle a,k\rangle$ belongs. Now we insert $W'$ between $\langle a,k_1\rangle$ and $\langle a,k_2\rangle$ to obtain the desired $\langle a,k_1\rangle W' \langle a,k_2\rangle$. In this case, the window $W$ is the empty string and its extended window is $x_1 y_1 y_2$, with $x_1=\langle a,k_1\rangle$ and $y_1=\langle a,k_2\rangle$. The proofs of Lemmas~\ref{lem:sourceW} and \ref{lem:targetW} apply verbatim, so still four chunk edits suffice. 

In case the four edits overlap a run $\langle a,k\rangle$, the analysis is the same without $\langle a,k_1\rangle$ or $\langle a,k_2\rangle$. A special case arises when the edits fuse two runs, as deleting $b$ in $\langle a,2\rangle b \langle a,2\rangle$. This case, which converts $\langle a,k_1\rangle W \langle a,k_2\rangle$ into $\langle a,k\rangle$, is the exact opposite of the one we analyzed in the previous paragraph, so by the symmetry of the edit operations it is still handled with 4 chunk edits.

\end{document}